\documentclass[a4paper,USenglish]{lipics}
\usepackage{microtype}%if unwanted, comment out or use option "draft"
\usepackage[T1]{fontenc}
\usepackage{ae,aecompl}
\usepackage{amssymb,amsmath}

\def\idrm#1{\ensuremath{\mathrm{#1}}}

\def\etal{\emph{et~al.}}

\newcommand{\no}[1]{}

\renewenvironment{proof}{\trivlist\item[]\emph{Proof}:}%
{\unskip\nobreak\hskip 1em plus 1fil\nobreak$\Box$
\parfillskip=0pt%
\endtrivlist}

\newenvironment{itemize*}%
  {\begin{itemize}%
    \setlength{\itemsep}{0pt}%
    \setlength{\parskip}{0pt}%
    \setlength{\parsep}{0pt}%
    \setlength{\topsep}{0pt}%
    \setlength{\partopsep}{0pt}%
  }%
  {\end{itemize}}%

\newcommand{\cT}{{\cal T}}

\newcommand{\tsa}{t_{SA}}

\newcommand{\ou}{\tilde{u}}
\newcommand{\ov}{\tilde{v}}
\newcommand{\ow}{\tilde{w}}
\newcommand{\tP}{\tilde{P}}

\newcommand{\shortver}[1]{}
\newcommand{\longver}[1]{#1}
\newcommand{\shlongver}[2]{#2}

\newcommand{\occ}{\idrm{occ}}

\newcommand{\eps}{\varepsilon}

\begin{document}
\title{Space-Efficient String Indexing for Wildcard Pattern Matching}
\author[1]{Moshe Lewenstein}
\author[2]{Yakov Nekrich}
\author[2]{Jeffrey Scott Vitter}
\affil[1]{Department of Computer Science, Bar-Ilan University} \affil[2]{Department of Electrical Engineering and Computer Science, University of Kansas}

\authorrunning{M. Lewenstein, Y. Nekrich and J. S. Vitter} %mandatory. First: Use abbreviated first/middle names. Second (only in severe cases): Use first author plus 'et. al.'

%\Copyright{Moshe Lewenstein, Yakov Nekrich and Jeffrey Scott Vitter}%mandatory. LIPIcs license is "CC-BY";  http://creativecommons.org/licenses/by/3.0/

\date{}
\maketitle

\begin{abstract}
In this paper we describe compressed indexes that support pattern matching queries for strings with wildcards.  For a constant size alphabet our data structure  uses $O(n\log^{\eps}n)$ bits for any $\eps>0$ and reports all $\occ$ occurrences of a wildcard string in $O(m+\sigma^g \cdot\mu(n)  + \occ)$ time, where $\mu(n)=o(\log\log\log n)$, $\sigma$ is the alphabet size, $m$ is the number of alphabet symbols and $g$ is the number of wildcard symbols in the query string. We also present an $O(n)$-bit index with $O((m+\sigma^g+\occ)\log^{\eps}n)$ query time and an $O(n(\log\log n)^2)$-bit index with $O((m+\sigma^g+\occ)\log\log n)$ query time.
These are the first non-trivial data  structures for this problem that need $o(n\log n)$ bits of space.
\end{abstract}

\section{Introduction}
\label{sec:intro}
In the string indexing problem, we pre-process a source string $T$, so that all occurrences of a query string $P$ in $T$ can be reported.  This is one of the most fundamental data structure problems. While handbook data structures, suffix arrays and suffix trees, can answer string matching queries efficiently, 
they store the source string $T$ in  $\Theta(\log n)$ bits of space per symbol. In situations when massive amounts of data must be indexed, the space usage  can become an issue.  Compressed indexes that use $o(\log n)$ or even $H_0$ bits per 
symbol, where $H_0$ denotes the zero-order entropy,  were studied extensively. We refer the reader to~\cite{MakinenN08} for a survey of results on compressed indexing.

In many scenarios we are interested in reporting all occurrences 
of strings  that resemble the query string $\tP$ but do not have to be identical to $\tP$. The problem of approximate pattern matching is important for biological applications and information retrieval and has received considerable attention~\cite{ColeGL04,LamSTY07,RahmanI07,TamWLY09,BilleGVV12,ChanLSTW11}. In this paper we consider a variant of the approximate pattern matching when the query string $\tP$ may contain wildcards (don't care symbols), and the wildcard symbol matches any alphabet symbol.

The standard indexing data structures can be used to answer wildcard pattern matching queries. A pattern $\tP$ with $g$ wildcard symbols matches $\sigma^g$ different patterns, where $\sigma$ denotes the size of the alphabet. We can generate all patterns that match $\tP$ and report all $\occ$ occurrences of these patterns (and hence all occurrence of $\tP$) in $O(m\cdot\sigma^g +\occ)$ time, where $m$ is the number of alphabet symbols. If the maximal number of wildcards in a query is bounded by $k$ ($k$-bounded indexing), we can  store a compressed trie with all possible combinations of $k$ wildcard symbols for every suffix. Then a query can be answered in $O(|\tP|+ \occ)$ time, but the total space usage is $O(n^{k+1})$ words of $\Theta(\log n)$ bits.

Cole \etal~\cite{ColeGL04} presented an elegant data structure for $k$-bounded indexing. Their solution needs $O(n\log^kn)$ words of space and answers wildcard queries in $O(m+2^g\log\log n +\occ)$ time. Very recently this has been improved in~\cite{LMRT:13} to $O(n\log^{k+\eps}n)$ bits of space with the same query time as Cole \etal~\cite{ColeGL04}. Bille~\etal~\cite{BilleGVV12} obtained another  trade-off: for any pre-defined $k$ and $\beta$, their $k$-bounded index uses $O(n\log n \log_{\beta}^{k-1}n)$ words and answers queries
in $O(m+\beta^g\log\log n +\occ)$ time. These indexes can provide fast answers to wildcard queries when the number of wildcards is small. However the space usage of the above data 
structures is high even when $k$ is a constant. For super-constant values of $k$ (for instance, when the maximal number of wildcards is bounded by $\log\log n$) the cost of storing the data structure may become prohibitive. 

Another line of research is the design of data structures that 
use linear or almost-linear space and support queries with 
an arbitrarily large number of wildcards.
Cole \etal~\cite{ColeGL04} describe a data structure
that uses $O(n\log n)$ words and answers queries in $O(m+\sigma^g\log\log n +\occ)$ time.  Iliopoulos and Rahman~\cite{RahmanI07} and Lam \etal~\cite{LamSTY07} describe linear-space indexes; however, their data structures need $\Theta(n)$ worst-case time to answer a query. Recently, Bille \etal~\cite{BilleGVV12} described an $O(n)$-words data structure that answers queries in $O(m+\sigma^g\log\log n +\occ)$ time.
\no{Chan \etal~\cite{ChanLSTW11} describe an index for a related 
$k$-mismatch problem: all substrings that match the query string with at most $k$ mismatches must be reported. Their data structure uses $O(n)$ bits and answers queries in  }

\begin{table}[t]
  \centering
 % \resizebox{\textwidth}{!}
  {\fontsize{12}{12}
    \begin{tabular}{|l|c|c|} \hline
    Ref. & Space Usage & Query Time \\ \hline
    %     & (in bits)   &            \\ \hline
         \cite{ColeGL04}  & $O(n\log n)$ words & $O(m+\sigma^g\log\log n +\occ)$ \\
         \cite{BilleGVV12} & $O(n)$ words & $O(m+\sigma^g\log\log n +\occ)$\\ \hline 
         New &  $O(n\log^{\eps} n\log \sigma)$ bits & $O(m+\sigma^g\sqrt{\log^{(3)}n} +\occ)$ \\
         New & $O(n(\log \log n)^2\log \sigma)$ bits & $O((m+\sigma^g +\occ)\log\log n)$\\
         New & $O(n\log\sigma)$ bits & $O((m+\sigma^g +\occ)\log^{\eps} n)$\\ \hline
       \end{tabular}
}
       \caption{Previous and new results on unbounded wildcard indexing; $m$ and $g$ denote the number of alphabet symbols and wildcards in the query pattern.}
       \label{tbl:results}
     \end{table}

When the amount of stored data is very large, even linear space usage can be undesirable. While numerous compressed indexes for exact pattern matching are known, there are no previously described data structures for wildcard indexing that use $o(n\log n)$ bits. In this paper we present  sublinear space indexes  for wildcard pattern matching. Our results are especially conspicuous when the alphabet size is constant.  Our first data structure uses $O(n\log^{\eps}n)$ bits and reports occurrences of a wildcard pattern in $O(m+\sigma^g\sqrt{\log^{(3)}n} +\occ)$ time\footnote{$\log^{(3)}n=\log\log \log n$.}; henceforth $\eps$ denotes an arbitrarily small positive constant. Thus we improve both the space usage and the query time of the previous best data structure~\cite{BilleGVV12}. The space usage can be further decreased at cost of slightly increasing the query time. We  describe two indexes that use $O(n)$ and $O(n(\log \log n)^2)$ bits of space; queries are supported in $O((m+\sigma^g +\occ)\log^{\eps} n)$ and $O((m+\sigma^g +\occ)\log \log n)$ time respectively. Previous and new results with worst-case efficient query times are listed in Table~\ref{tbl:results}. 

In this paper we assume, unless specified otherwise, that the 
alphabet size is a constant.
But our techniques are also relevant for  the case when the alphabet size is arbitrarily large. We can obtain an $O(n\log \sigma)$-bit data structure that answers queries in $O((m+\sigma^g+\occ)\log_{\sigma}^{\eps}n)$ time. We can also obtain an $O(n\log n)$-bit data structure that supports queries in 
$O(m+\sigma^g+\occ)$ time if $\sigma\ge \log\log n$. Other interesting trade-offs are 
possible and will be described in the full version of this paper. 

In Section~\ref{sec:prelim}, we recall some results related to 
compressed suffix trees and suffix arrays and compressed data structures for a set of integers. We also define the unrooted LCP queries, introduced in Cole~\etal~\cite{ColeGL04}, that are the main tool in all currently known efficient structures for wildcard indexing.  In Section~\ref{sec:smallset} we describe data structures that answer unrooted LCP queries on a small subtree 
of the suffix tree. Our data structures need only a small number of additional bits if the (compressed) suffix tree and suffix array of the source text are available. 
In Section~\ref{sec:lessspace}, we describe compact data structures that answer LCP queries and wildcard pattern matching queries on an arbitrarily large suffix tree. These data structures are based on  a subdivision of suffix tree nodes into small subtrees. 
In Sections~\ref{sec:case0},~\ref{lemma:case1}, and \ref{sec:case2} we show how we can speed-up the data structures from~\cite{ColeGL04}, \cite{BilleGVV12} and retain $o(n\log n)$ 
space usage. The main component of our improvement is  a method for  processing batches of unrooted LCP queries. In previous works~\cite{ColeGL04,BilleGVV12} LCP queries were answered one-by-one.

\section{Preliminaries}
\label{sec:prelim}
{\bf Unrooted LCP Queries.}
% Let $\cT$ be  a compressed trie over a set of strings $S$. 
% In this paper $S$ will always be a set of suffixes of some text $T$.  
In this paper $s_1\circ s_2$ denotes the concatenation of strings $s_1$ and $s_2$ and $\cT$ denotes the  suffix tree of the source text.
A string $str(v,u)$ is obtained by concatenating labels of all edges on the path from $v$ to $u$ and  $str(u)=str(v_r,u)$ for the root node $v_r$ of $\cT$.
A \emph{location} on a suffix tree $\cT$ is an arbitrary position on an edge of $\cT$; a location on an edge $(v,u)$ 
can be uniquely identified by specifying the edge $(u,v)$ and the offset from the upper node of $(u,v)$. We can straightforwardly extend the definitions 
of $str(\ov,\ou)$ and $str(\ou)$ to arbitrary locations $\ou$ and $\ov$. 
The unrooted LCP query $(v,P)$, defined in~\cite{ColeGL04}, asks for the lowest
descendant location $\ou$ of a node $v$, such that $str(v,\ou)$ is a prefix of a string $P$.
 Thus an unrooted LCP query provides the answer to the following question: if we were to search  for a pattern $P$ in a subtree with root $v$, where would the search end? While we can obviously answer this question in $O(|P|)$ time by traversing the trie starting at $v$, faster solutions are also possible.

As in the previous works~\cite{ColeGL04,BilleGVV12}, we consider the following two-stage scenario for answering queries: during the first stage an arbitrary string $P$ is pre-processed in $O(|P|)$ time; during the second stage, we answer queries 
$(u, P_j)$ for any suffix $P_j$ of $P$ and any $u\in \cT$.  
Cole \etal~\cite{ColeGL04} described an $O(n\log^2 n)$-bit data structure that answers unrooted LCP queries in $O(\log\log n)$ time. Bille~\etal~\cite{BilleGVV12} improved the space usage to linear ($O(n\log n)$ bits). \\
% \begin{lemma}[\cite{BilleGVV12}]
% \label{lemma:unrooted1}
% Suppose that $S$ is a set containing some suffixes of a text $T$ and $\cT$ is the compressed trie of $S$. Let $|T|=n$ and $|S|=f$.
% After pre-processing a pattern $P$ in $O(|P|)$ time, we can answer unrooted 
% LCP queries for any suffix of $P$ and for any node $v\in\cT$ in $O(\min(\log f, \log \log n)$ time. 
% The underlying data structure uses $O(f)$ additional space. 
% \end{lemma}
% The pre-processing stage is the same in ~\cite{ColeGL04}, ~\cite{BilleGVV12}, and in the data structures that will be used in this  paper: during the pre-processing stage we find 
% the position of every suffix $P'$ of $P$ in the sorted set of all suffixes of $T$.

{\bf Compressed Suffix Arrays and Suffix Trees.}
The suffix array $SA$ for a text $T$ contains starting positions of $T$'s suffixes sorted in lexicographic order: $SA[i]=k$ 
if the suffix $T[k..n]$ is the $k$-th smallest suffix of the text $T$. We will say that $i$ is the rank of the suffix $T[k..n]$. An inverse suffix array stores information about lexicographic order of suffixes: $SA^{-1}[k]=i$ iff $SA[i]=k$. 
We will say that a data structure provides a suffix array functionality in time $\tsa$ if it enables us to compute $SA[i]$ and $SA^{-1}[k]$ for any $1 \le i,k\le n$ in $O(\tsa)$ time.
A number of compressed data structures provide suffix array functionality in little time. 
% For example, the compressed suffix array of Sadakane~\cite{ } achieves $O(nH_0(T))$ space and $\tsa=O(\log^{\eps}n)$ or $O(nH_0(T)\log\log n)$ space and 
% $\tsa=O(\log\log n)$. The data structure of ??~\cite{??} uses $nH_k(T)+??$ space and provides suffix array functionality in 
% $O( )$ time. 
\setlength{\tabcolsep}{3em}
\begin{lemma}
\label{lemma:csa}
If the alphabet size $\sigma=O(1)$, the following trade-offs for space usage $s(n)$ and $\tsa$ are possible:
(a) $s(n)=O((1/\eps)n)$  and $\tsa(n)=O(\log^{\eps}n)$, or 
(b) $s(n)=O(n\log\log n)$ and  $\tsa(n)=O(\log\log n)$, or
(c) $s(n)=O(n\log^{\eps}n)$  and $\tsa(n)=O(1)$
% \begin{tabular}{l@{\hspace{.2cm}}l@{\hspace{1em}}l}
% (a) & $s(n)=O((1/\eps)n)$, & $\tsa(n)=O(\log^{\eps}n)$ \\
% (b) & $s(n)=O(n\log\log n)$, &  $\tsa(n)=O(\log\log n)$\\
% (c) & $s(n)=O(n\log^{\eps}n)$,  & $\tsa(n)=O(1)$\\
% \end{tabular}
% \\
for any constant $\eps>0$
\end{lemma}
\begin{proof}
  Result (a) is shown in~\cite{Sadakane00} and results (b), (c) are from~\cite{Rao02}
\end{proof}
If $SA[t]= f$ the function $\Psi^i(t)$ computes the position of the suffix $T[f+i..n]$ in the suffix array. This function can be computed in $O(\tsa)$ time as $SA^{-1}[SA[t]+i]$. 
Let the string depth of a node $v\in \cT$ be the length $str(v)$.
If the suffix array functionality is available, we can  store the suffix tree in $O(n)$ additional bits, so that the string depth of any node $v$ can be computed in $O(\tsa)$ time~\cite{Sadakane07,FischerMN09,RussoNO11}. 

% Let the rank of a string $P$ (in the text $T$) be the number of $T$'s suffixes that are lexicographically smaller or equal to 
% $T$. That is, the rank of $P$ in $T$ is its position in the suffix array of $T$. Using $O(n\log\sigma)$ bits\footnote{A better estimation of space usage is possible, but it is not relevant for results that we present in this paper.}, we can find the rank of any pattern $P$ and the ranks of all $P$'s suffixes in $O(|P|)$ time~\cite{BelazzouguiN11}. 

Using $O(n)$ additional bits, we can process a string $P$ in $O(|P|t_{SA})$ time and find for any suffix $P^j=P[j..|P|]$ of $P$:
(i) the rank $r_j$ of $P^j$ in $T$ and (ii) the longest common prefix (LCP) of $P^j$ and the suffixes $SA[r_j]$, $SA[r_j+1]$ of $T$. Using McCreight's procedure for inserting a new string into a generalized suffix tree, we  find the locations where suffixes
of $P$ must be inserted into $\cT$: first we traverse the suffix tree starting at the root and find the location corresponding to $P[1..|P|]$ in the suffix tree; then we find locations of $P[2..|P|]$, $\ldots$, $P[|P|-1..|P|]$, $P[|P|]$ by following the suffix links. Next, we compute the string depths of these locations. The total time needed to find the locations and their depths in a compressed suffix tree is 
$O(|P|\tsa)$.  When the rank $r_j$  of $P^j$ and LCPs of $P^j$ and its neighbors are known, we can use this information to compute 
the LCP of $P^j$ and any suffix $SA[q]$ in $O(\tsa)$ time: 
if $q<r_j$, $LCP(P^j,SA[q])$ is the minimum of $LCP(P^j,SA[r_j])$ and $LCP(SA[r_j],SA[q])$; the case $q>r_j$ is symmetric. 
Sadakane~\cite{Sadakane07} showed how to compute $LCP(SA[r_j],SA[q])$ in $O(\tsa)$ 
time. Hence, we can compute the LCP for any two suffixes of $P$ 
and $T$ in $O(\tsa)$ time after $O(|P|\tsa)$ pre-processing time.

{\bf Heavy Path Decomposition.}
Let $\cT$ be an arbitrary tree. We can decompose $\cT$ into disjoint root-to-leaf paths, called \emph{heavy paths}. 
If an internal node $u\in \cT$ is on a heavy path $p$, then its heaviest child $u_i$ (that is, the child with  the greatest number of leaf descendants) is also on $p$. 
If the child $u_j$ of $u$ is not on $p$, then $u$ has at least twice as many leaf descendants as $u$. 
Therefore the heavy-path decomposition of $\cT$ guarantees that  \emph{any} root-to-leaf path in $\cT$ intersects with at most $\log n$ heavy paths; we refer to~\cite{HarelT84} for details.
% For a heavy path $h_j$, we will select an arbitrary string $str(u,v_l)$ 
% such that $u$ is the highest node in $h_j$ and $v_l$ is an arbitrary leaf descendant of the lowest node in $h_j$; we will say that $str(u,v_j)$ is a \emph{representative} of $h_j$. 

{\bf Searching in a Small Set.} 
We can search in a set with a poly-logarithmic number of elements
using the data structure called an atomic heap~\cite{FW94}. An atomic heap on a set of integers~$S$, $|S|=\log^{O(1)} n$, uses linear space and enables us to find for any integer $q$ the largest $e\in S$ such that $e\le q$ (respectively, the smallest $e\in S$ such that $e\ge q$) in $O(1)$ time. Using the result of Grossi~\etal~\cite{GrossiORR09}, we can search in a small set  using small 
additional space and only one access to elements of~$S$. 
\begin{lemma}[\cite{GrossiORR09}, Lemma 3.3]
\label{lemma:small}
Suppose that $|S|=\log^{O(1)}n$ and $e\le n$ for any $e\in S$.
There exists a data structure $D$ that uses $O(|S|\log\log n)$ 
additional bits and answers predecessor and successor queries 
on $S$ in $O(1)$ time. When a query is answered, only one element $e'\in S$ needs to be accessed.  
\end{lemma}

\section{Unrooted LCP Queries on Small Sets }
\label{sec:smallset}
 
% Cole \etal.~\cite{ColeGL04} 
% described an $O(n\log n)$ space data structure that supports the following 
% queries: we can pre-process an arbitrary string $P$ in $O(|P|)$ time 
% so that later we can answer unrooted LCP queries for any node $v\in \cT$ 
% and for any suffix $P'$ of $P$ in $O(\log \log n)$ time.  Bille~\etal.~\cite{BilleGVV12} improved the space usage to linear. 
% They also showed that logarithmic query time can be achieved.
 
In this section we describe compact data structures that answer  LCP queries on a small set of suffixes. We consider a set $S$ that contains a poly-logarithmic 
number of consecutive suffixes from the suffix array of $S$. Our data structure supports queries 
of the form $(u_0,P)$ where $u_0\in\cT_0$ and $\cT_0$ is a subtree of the suffix tree $\cT$ induced by suffixes from $S$;
the query answer is the lowest location  $\ov\in\cT_0$ below $\ou$, 
such that $str(u_0,\ov_0)$ is a prefix of $P$. These data structures are an important building block of data structures that will be constructed in the following sections and a key to space-saving solution: we will show in section~\ref{sec:lessspace} how a suffix tree can be divided into small subtrees. In this section we show how unrooted LCP queries can be supported on such small subtrees.
The main idea is to keep the (ranks of) suffixes  in succinct predecessor data structures that need $O(\log \log n)$ additional bits per element; we do not have to store the  ranks in these data structures because they can be retrieved in $O(\tsa)$ time using the (compressed) suffix tree and the (compressed) suffix array. Thus we can answer unrooted LCP queries on $\cT_0$ using $O((\log\log n)^2)$ bits per suffix. We assume in the rest of this section that $S$ contains $f=O(\log^3n)$ consecutive suffixes and $\cT_0$ 
is a subtree of the suffix tree induced by suffixes from $S$.

\begin{lemma}\label{lemma:unrooted2}
%Suppose that a set $S$ contains $f$ consecutive  suffixes of a text $T$ for  $f=O(\log^3 n)$ and $|T|=n$. 
There exists a data structure that uses $O(f (\log \log n)^2)$ additional bits of space and answers unrooted LCP queries on $\cT_0$ in $O(1)$ time. We assume that our data structure can access the suffix tree of $T$, the suffix array of $T$, the inverse suffix array of $T$,  and a universal look-up table of size $O(n^g)$ for an arbitrarily small positive constant $g$. 
\end{lemma}
\begin{proof}
Let $\cT_0$ denote the part of the suffix tree induced by suffixes in $S$.
We apply the heavy path decomposition to nodes of $\cT_0$. 
 Let $S(u)$ denote the set that contains all strings $str(w,v_l)$ for the parent $w$ of $u$ and all leaf descendants $v_l$ of $u$. 
We remark that all elements of $S(u)$ are suffixes of $T$. 
The global rank of a suffix $Suf$ is its position in the suffix array of $T$. 
Let $R(u)$ denote the set of global ranks of all suffixes in $S(u)$. 
For every node $u\in \cT_0$ and each of its children $u_i$ that are not on the same heavy path as $u$, we store  a data structure $D(u_i)$. $D(u_i)$ answers predecessor queries on $R(u_i)$. It is not necessary to store the set $R(u)$ itself: 
an arbitrary element of $R(u)$ can be accessed using the  functionality provided by the suffix array. 
Suppose that the global rank of the suffix corresponding to $str(w,v_p)$, where $v_p$ is the $p$-th leaf descendant of $S(u)$, should be computed.   Since we can access the suffix tree, we can find the rank $r_1$ of the suffix that ends in the leaf $v_p$.
Then the suffix corresponding to $str(w,v_p)$ has rank $SA[SA^{-1}[r_1]+depth(w)]$ where $depth(w)$ is the string depth of the 
node $w$ in the global suffix tree. 
By Lemma~\ref{lemma:small}, $D(u_i)$ can be stored in 
$O(|S(u_i)|\log \log n)$ bits and answer predecessor queries 
in $O(1)$ time.  
% contains representatives of all strings $str(u_i,v)$ for any $v$ that is a leaf in $\cT_0$ and a descendant of $u_i$. 
% Every representative is a suffix of $T$; hence, it can be identified by the index of that suffix in the lexicographic order. Since the total number of suffixes in $D(u)$ is
% $f=O(\log n)$, we can search in each data structure $D(u)$ in $O(1)$ time using an atomic heap data structure~\cite{FW94}.
 The total number of elements in all $D(u)$ is $O(f \log f)=O(f\log\log n)$. Thus all $D(u)$ need  $O(f(\log\log n)^2)$ bits or $o(f)$ words of $\log n$ bits. 
For every heavy path $h_j$ on $\cT_0$ we keep a data structure 
$H_j$ that contains the depths of all nodes. $H_j$ is also implemented as described in Lemma~\ref{lemma:small} and uses 
$O(\log \log n)$ bits per node. 

The search for an LCP in $\cT_0$ is organized in the same way as in ~\cite{ColeGL04}. 
To answer a query $(u,P_j)$, $u\in \cT_0$, we start by finding $l_0=lcp(P_j,SA[r])$, where $r$ is the rank of the suffix  that starts at $u$ and ends in the leaf $v_h$, such that $u$ and $v_h$ are on the same heavy path.  
Let $u'$ denote the lowest node of depth  $d_1\le depth(u)+l_0$ that is on the same heavy path $h_0$ in $\cT_0$ as $u$. 
If $d_1\not=depth(u)+l_0$, then $u'$ is the answer to our query. If $d_1=depth(u)+l_0$ and $u'$ is a leaf, then again $u'$ is the answer to our query. If $d_1=depth(u)+l_0$ and $u'$ is not a leaf, we identify the child $u_j$ of $u'$ that is labelled with $P_j[d_1+1]$. If such a child does not exist, then again $u'$ is the answer.  Otherwise, we find the rank $r'$ of $P'_j=P_j[d_1+1..|P_j|]$.
Using $D(u_j)$, we find the predecessor and the successor of $r'$ in $S(u_j)$. 

Let $S_l$ and $S_r$ denote the corresponding suffixes 
of $D(u_j)$.  We can compute $l_l=lcp(P'_j,S_l)$ and $l_r=lcp(P'_j,S_r)$. Suppose that $l_l\ge l_r$. Let $u_l$ be the node of depth at most $depth(u_j)+l_j$ on the path from $u_j$ to the leaf $l_l$ containing $S_l$.
The node $u_l$, that can be found by answering an appropriate level ancestor query for $l_l$, is the answer to the original  LCP query.  The case when $l_r>l_l$ is handled in the same way.  
\end{proof}

In the following two Lemmas we extend the result of Lemma~\ref{lemma:unrooted2} to the situation  when the data structure is  stored in compressed form. We assume that 
we can  compute $SA[i]$, $SA^{-1}[i]$ for any $i$, $1\le i \le n$, in $O(\tsa)$ time; we also assume that compressed suffix tree with functionality described in Section~\ref{sec:prelim} is available. Only additional bits 
necessary to support queries on $\cT_0$ are counted. 
\begin{lemma}\label{lemma:unrooted3}
%Suppose that a set $S$ contains $f$ consecutive  suffixes of a text $T$ for  $f=O(\log^3 n)$ and $|T|=n$. 
There exists a data structure that uses $O(f (\log \log n)^3)$ additional bits of space and answers unrooted LCP queries on $\cT_0$ in $O(\tsa)$ time. 
Our data structure  uses a universal look-up table of size $O(n^g)$ for an arbitrarily small positive constant $g$.
\end{lemma}
\begin{proof}
We use the same data structure as in the proof of Lemma~\ref{lemma:unrooted3}, but $SA[SA^{-1}[r_1]+depth(w)]$ and $depth(u)$ are computed in $O(\tsa)$ time.  It is not necessary to store $\cT$. Information about the heavy path decomposition of $\cT_0$ can be stored in $O(f)$ bits. \shlongver{We show how this can be done in~\cite{arxivver}.}{We will show how this can be done in Appendix~\ref{sec:hpath}.} Data structures $H_i$ need $O(\log\log n)$ bits per node. Since queries on $H_j$ and $D(u)$ are answered in $O(\tsa)$ time, an unrooted LCP 
query is also answered in $O(\tsa)$ time. 
\end{proof}
The following Lemma is  proved in\shlongver{~\cite{arxivver}.}{Appendix~\ref{sec:hpath}.}
\begin{lemma}\label{lemma:unrooted4}
%Suppose that a set $S$ contains $f$ consecutive  suffixes of a text $T$ for  $f=O(\log^3 n)$ and $|T|=n$. 
There exists a data structure that uses $O(f)$ additional bits of space and answers unrooted LCP queries on $\cT_0$ in $O((\tsa(\log\log\log n))$ time. 
Our data structure  uses a universal look-up table of size $O(n^g)$ for an arbitrarily small positive constant $g$.
\end{lemma}

\section{Wildcard Pattern Queries in Less Space}
\label{sec:lessspace}
Now we are ready to describe the compact data structure for wildcard indexing. Our approach is as follows. We divide the suffix tree $\cT$ into subtrees, so that each subtree has a poly-logarithmic number of nodes and results of Section~\ref{sec:smallset} can be applied to each subtree. We also keep a  tree $\cT_m$ that has one representative node for each subtree and stores information about positions of small subtrees in $\cT$. Unrooted LCP queries 
are answered in two steps. First, we identify the small subtree that contains the answer using data structures on $\cT_m$. 
Then we search in the small subtree using the data structure of Section~\ref{sec:smallset}. 
We select the size of subtrees so that $\cT_m$ and data structures for $\cT_m$ use $O(n)$ bits. A detailed description of our data structure is given below. 

{\bf Data Structure.}
Let $\tau=\sigma\log^2 n$. 
We visit all leaves of the suffix tree $\cT$ in left-to-right order and mark every $\tau$-th leaf. We visit all internal nodes of $\cT$ in bottom-to-top order and  mark each node $u$ such that at least two children of $u$ have marked descendants. Finally the root node is also marked.

We divide the nodes of the suffix tree into groups as follows.
Let $u$ be a marked internal node, such that all its non-leaf descendants are unmarked. Each child $u_i$ of $u$ contains at most one marked leaf (because otherwise the subtree rooted at $u_i$ would contain marked internal nodes). The subtrees rooted  at children $u_i, \ldots, u_d$ of $u$ are distributed among groups $G_j(u)$. We select indices $i_1=1$, $i_2$, $\ldots$, $i_t=m$ such that exactly one node among   $u_{i_j}$, $\ldots$, $u_{i_{j+1}-1}$ has a marked 
leaf descendant. For each $j$, $1\le j< t$, all nodes in the subtrees of $u_{i_j},\ldots, u_{i_{j+1}-1}$ are assigned to group 
$G_j(u)$. Every $G_j(u)$ contains $O(\tau)$ nodes. 
Now suppose that a marked node $u$ has marked descendants. 
We divide the children of $u$ into groups $G(u,v)$
 such that 
exactly one child $u_i$ of $u$ in each  $G(u,v)$ has exactly one direct marked descendant. That is, in every $G(u,v)$ there is  exactly one child $u_i$ of $u$ satisfying one of the following two conditions: (i) $u_i$ is marked (in this case $u_i$ is assigned to the group $G(u,u_i)$)
or (ii) $u_i$ has exactly one marked descendant $v$ such that 
there are no other marked nodes between $u_i$ and $v$. 
The group $G(u,v)$ also contains all nodes that are descendants of $u_i$ but are not proper descendants of $v$. To make nodes of 
$G(u,v)$ a subtree, we also include $u$ into $G(u,v)$. 
The number of nodes in $G(u,v)$ is also bounded by $O(\tau)$.

Each node $w\in\cT$ belongs to some group $G_j(u)$ or $G(v,u)$. 
The total number of groups is $O(n/\tau)$ because each group 
can be associated with one marked node. Since every $G_j(u)$ 
is a subtree, we can answer unrooted LCP queries on the nodes 
(and locations) of $G_j(u)$ implemented according to Lemma~\ref{lemma:unrooted3}. 
Furthermore we divide every $G(v,u)$ into two overlapping subgroups: $G_l(v,u)$ contains all nodes of $G(v,u)$ that are 
on the path from $v$ to $u$ or to the left of this path; $G_r(v,u)$ contains all nodes of $G(v,u)$ that are on the path from $v$ to $u$ or to the right of this path. 
We also add the leftmost and rightmost leaf descendants of the 
node $u$, where $u$ is the marked node in $G(v,u)$,  to
$G_l(v,u)$ and $G_r(v,u)$ respectively. The leaves in each group $G_l(v,u)$ and $G_r(v,u)$ correspond to $\tau$ consecutive suffixes. Therefore we can answer unrooted LCP queries on $G_l(u,v)$ and $G_r(u,v)$ using Lemmas~\ref{lemma:unrooted3} or~\ref{lemma:unrooted4}. The answer to an unrooted LCP query on $G(u,v)$ can be obtained 
from answers to the same query on $G_l(u,v)$ and $G_r(u,v)$.  
The data structures for unrooted LCP queries 
on $G_j(u)$, $G_l(u,v)$ and $G_r(u,v)$ will be denoted $D_j(u)$, 
$D_l(u,v)$ and $D_r(u,v)$ respectively. Each node belongs to at most two groups; therefore all group data structures need $O(n)$ bits of space.

The nodes of the suffix tree are stored in compressed form described in Section~\ref{sec:prelim}.  The depth and the string depth of any node can be computed 
in $O(\tsa)$ time. We can also pre-process an arbitrary pattern in $O(|P|\tsa)$ time, so that the LCP of any suffixes $P[j..|P|]$ and $T[i..n]$ can be found in $O(\tsa)$ time.

Moreover, we keep all suffixes that are stored in marked leaves of the suffix tree in a compressed trie $\cT_m$. Nodes of $\cT_m$ correspond to marked nodes of $\cT$. \longver{We keep the data structure 
of Lemma~\ref{lemma:unrootedlarge}  that supports unrooted LCP queries on the nodes of $\cT_m$ in $O(\log \log n)$ time.
This data structure uses $O((n/\tau)\log^2 n)=O(n/\sigma)$ bits.}
\shortver{Unrooted LCP queries on $\cT_m$ can be answered in $O(\log \log n)$ time using $O((n/\tau)\log^2 n)=O(n/\sigma)$ bits; see Lemma 11 in~\cite{arxivver}.} 
 
In every node of $\cT_m$ we store a pointer to the corresponding marked node of $\cT$. We also keep a bit vector  $B$ that keeps data about marked and unmarked nodes of $\cT$; the order of nodes is determined by a pre-order traversal of $\cT$. The $i$-th entry $B[i]$ is set to $1$ if the $i$-th node (in pre-order traversal) is marked, otherwise $B[i]$ is set to $0$.
Using $o(n)$ additional bits, we can compute the number of preceding 1's for any position in $B$ 
in $O(1)$ time~\cite{Munro96}. 
Hence for any node $u\in \cT$, we can find the number of marked nodes that precede $u$ in the pre-order traversal of $\cT$. We also store an array 
$A_m$; the $i$-th entry of $A_m$ contains a pointer to the node 
of $\cT_m$ that corresponds to the $i$-th marked node in $\cT$.
Using $B$ and $A_m$, we can find the node of $\cT_m$ that 
corresponds to a given marked node of $\cT$ in $O(1)$ time. 
We will also need another  data structure to facilitate the navigation between marked nodes and its children. For every marked node $u$ with marked internal descendants and for all groups $G(u,v)$, we store the first character on the label of the edge from $u$ to its leftmost child $u_i\in G(u,v)$ in a predecessor data structure.

{\bf Queries.}
Consider an unrooted LCP query $(u,P)$. If $u$ is marked, we 
find the lowest marked descendant $u'$ of $u$, such that 
$str(u,u')$ is a prefix of $P$. We find the child $u_i$ of $u'$ such that the edge from $u'$ to $u_i$ is labelled with a string 
$s_i$ and $str(u,u')\circ s_i$ is a prefix of $P$. Then we use 
the data structure $D_j(u)$ (respectively $D_l(u,w)$ and $D_r(u,w)$) for the subtree that contains $u_i$ and  answer an unrooted LCP query $(u_i,P')$ for $P'$ satisfying 
$str(u,u')\circ s_i\circ P'=P$. The answer to the latter query provides the answer to the original query $(u,P)$. 
If $u$ is unmarked, we start by  answering the query $(u,P)$ using the data structure for the  group that contains $u$.
If the answer is an unmarked node $u_1$ (or a location $\ou_1$ 
on an edge that starts in an unmarked node), 
then $u_1$ (respectively $\ou_1$) is the answer to our query. If $u_1$ is marked, we answer the query $(u_1,P_1)$, 
where $P_1$ is the remaining suffix of $P$, as described above. Again we obtain the answer to the original query $(u,P)$. 
%This procedure can be easily extended to the case when a query $(\ou,P)$ must be answered and $\ou$ denotes a location
%on some edge $(u,w)$ of the suffix tree. 

We can report all occurrences of  $\tP=\phi P_1\phi P_2\ldots \phi P_d$ by answering at most $\sigma^d$ unrooted LCP queries and $\sigma^d$ accesses to the compressed suffix tree. For all alphabet symbols $a$ we find the location of the pattern $aP_1$ by answering a wildcard LCP query. 
For each symbol $a$, such that the location $\ou_a$ of $aP$ in $\cT$ was found, we continue as follows. 
If $\ou_a$ is a position on an edge $(u_a,u_a')$, we check whether the remaining part of the edge label equals $aP_2'$ for some symbol $a$ and a prefix $P_2'$ of $P_2$. If this is the case, we answer a query $(u_a',P_2'')$ where $P_2''$ satisfies $P_2=P_2'\circ P_2''$. If $\ou_a$ is a node, we find the loci of patterns $str(\ou_a)\circ xP_2$, where $x$ denotes any alphabet symbol, as described above. We proceed in the same way until the loci of all $x_1P_1\ldots x_mP_m$ for any alphabet symbol $x_i$ are found. This approach can be straightforwardly extended to  reporting occurrences of a general wildcard expression $\tP=\phi ^{k_1}P_1\phi ^{k_2}P_2\ldots \phi ^{k_d}P_d$, where $\phi ^{k_i}$ denotes an arbitrary sequence of $k_i$ alphabet symbols and
$k_i\ge 0$ for $1\le i\le d$. 
\begin{theorem}
\label{theor:comp1}
There exists an $O(n+s_{small}n)$-bit data structure that reports all $\occ$ occurrences of a wildcard pattern  $\phi ^{k_1}P_1\phi ^{k_2}P_2\ldots \phi ^{k_d}P_d$ in $O(\sum_{i=1}^d|P_i|\tsa +\sigma^g t_{small}(n) +\occ\cdot\tsa)$ time, where  $g=\sum_{i=1}^mk_i$; $s_{small}$ and $t_{small}$ denote the average space usage and query time of the data structures described in Lemmas~\ref{lemma:unrooted2} or ~\ref{lemma:unrooted3}. 
\end{theorem}

Two interesting corollaries of this result are the following indexes. We use the same notation as in Theorem~\ref{theor:comp1}. If we combine Lemma~\ref{lemma:csa}, (a) with Lemma~\ref{lemma:unrooted4} we get $t_{small}=O(\log^{\eps}n)$ and $s_{small}=O(1)$ (the query time $O(\log^{\eps}n\log^{(3)}n)$ can be simplified to $O(\log^{\eps}n)$ by replacing $\eps$ with some $\eps'<\eps$). If we plug in this result into Theorem~\ref{theor:comp1}, we obtain our first main data structure.
\begin{corollary}
  \label{cor:compconst1}
There exists an $O(n)$-bit data structure that answers wildcard 
pattern matching queries in $O((\sum_{i=1}^d|P_i| +\sigma^g+\occ)\log^{\eps}n)$  time. 
\end{corollary}
We remark that the result of Corollary~\ref{cor:compconst1} can be also extended to the case of an arbitrarily large alphabet. 
In this case the index uses $O(n\log \sigma)$ bits and queries are answered in $(\sum_{i=1}^d|P_i| +\sigma^g+\occ)\log_{\sigma}^{\eps}n)$ time. This variant can be obtained by using the suffix array of Grossi \etal~\cite{GrossiV05}; the compressed suffix tree uses $O(n\log\sigma)$ bits in this case.

If we combine Lemma~\ref{lemma:csa}, (b) with Lemma~\ref{lemma:unrooted4} and plug in the result into Theorem~\ref{theor:comp1}, we obtain our second main data structure.
\begin{corollary}
  \label{cor:compconst2}
There exists an $O(n(\log\log n)^2)$-bit data structure that answers wildcard pattern matching queries in $O((\sum_{i=1}^d|P_i|+\sigma^g + \occ)\log\log n)$ time. 
\end{corollary}

\section{LCP  Queries for Patterns with Wildcards, $\sigma=\log \log n$}
\label{sec:case0}
In the remaining part of this paper we describe  faster solutions that use  linear or sublinear space. In sections~\ref{sec:case0} and~\ref{sec:case1} we describe an $O(n\log n)$-bit data structure for $\sigma\ge\log \log n$. 
In section~\ref{sec:case2} we use a more technically involved variant of the same approach to obtain fast solutions for $\sigma<\log\log n$.
 
In this section we will show how to answer a batch of LCP queries  called wildcard LCP queries.
A wildcard LCP query $(u,\phi P)$ returns the loci of  $str(u)\circ aP$ in the suffix tree of a source text $T$ 
for all $a\in \Sigma$  such that $str(u)\circ aP$ occurs in $T$. As before, we assume that we can preprocess some pattern 
$\overline{P}$ in $O(\overline{P})$ time; then, queries $(u,P)$ where $P$ is a suffix of $\overline{P}$ are answered. The pre-processing is the same as in Section~\ref{sec:smallset}.

A leaf descendant $v_l$ of a node $u$ is a light descendant of $u$ if $v_l$ and $u$ are not on the same heavy path.
A wildcard tree $\cT_u$ for a node $u$ is a compressed trie that contains all strings $s$ satisfying
$a\circ s=str(u,v_l)$ for some symbol $a$ and some light leaf descendant $v_l$ of $u$. 
The main idea of our approach is to augment the suffix tree $\cT$ with wilcard 
trees in order to accelerate the search. To avoid logarithmic  increase in space 
usage, only selected nodes of wilcard trees will be stored. We explain our method for the  case $\sigma=\log\log n$. 

Let $\tau=\sigma\log^2 n$. We mark the nodes of the suffix tree in the same way as described in Section~\ref{sec:lessspace}. Every 
$\tau$-th leaf of $\cT$, each internal node with at least two children that have marked descendants, and the root of $\cT$ are marked. The nodes of $\cT$ will be called the \emph{alphabet nodes}. We also store selected nodes from wildcard trees, further called \emph{wildcard nodes}.   A truncated wildcard tree $\cT_u$ is a compressed trie containing all strings $s$, such that $a\circ s=str(u,v_l)$ for some marked  light leaf descendant  $v_l$ of $u$.  Each leaf-to-root path intersects  $O(\log n)$ heavy paths. Therefore each marked leaf occurs in $O(\log n)$ truncated wildcard trees. Hence 
the total number of wildcard nodes  is $O((n/\tau)\log n)$. Every node in each truncated wildcard tree contains pointers to some alphabet nodes or locations on edges between alphabet nodes. Suppose that  a node $v$ is in a wildcard subtree $\cT_w$, the parent of $\cT_w$ is some node $w$, and the label of $v$ in $\cT_w$ is $s$. 
For every symbol $a$ such that $s_a=str(w)\circ a\circ s$ occurs in the source text, we store a pointer from $u$ to the location $u_a$ of $s_a$. 
% If $s_a$ is a proper prefix of $str(u_a)$ we also store a pointer to the parent $u'_a$ of $u_a$ and the difference between the length of $s_a$ and the string depth of $u'_a$. In other words if $s_a$ is shorter then $str(u_a)$, we store information about the position on the edge $(u'_a,u_a)$ such that the concatenation of all edge labels from the root to that position is equal to $s_a$. 
The total number of pointers is equal to $O(n\log n (\sigma/\tau))$. 
We distribute alphabet nodes into groups $G_j(u)$ and $G(v,u)$ 
as described in Section~\ref{sec:lessspace}; data structures $D_j(u)$, $D_l(v,u)$, and $D_r(v,u)$ are also defined in the same way as in Section~\ref{sec:lessspace}. Every pointer from a wildcard node to an alphabet node $w$ (or edge $(u,w)$) contains a reference to the group that contains $w$. \longver{Moreover, both alphabet and wildcard nodes of our extended suffix tree are kept in the data structure of Lemma~\ref{lemma:unrootedlarge} that   answers unrooted LCP queries in $O(\log\log n)$ time.}
\shortver{Moreover, both alphabet and wildcard nodes of our extended suffix tree are kept in the data structure, described in  Lemma 11 in~\cite{arxivver}, that   answers unrooted LCP queries in $O(\log\log n)$ time.}

{\bf Queries.}
Suppose that a wildcard LCP query $(u,\phi P)$ must be answered. Let $a_h$ be the first symbol in $str(u,u_h)$, where $u_h$ is the child of $u$ that is on the same heavy path. We answer a query $a_h\circ P$ in $O(\log \log n)$ time using the result of~\cite{BilleGVV12}. Next, we must find the locus nodes of all patterns $a_j\circ P$, $a_j\not= a_h$. 
We answer an LCP query $P$ in the truncated wildcard tree $\cT_u$ of the node $u$. Let $w$ denote the node where the search  for $P$ 
in $\cT_u$ ends and let $w_r$ denote the root node of $\cT_u$.  The node $w$ can also be found in $O(\log \log n)$ time.  
\begin{description}
\item[1.]
Suppose that $str(w_r,w)=P$. 
We follow pointers from $w$ to alphabet nodes $w_1$, $\ldots$, $w_{\sigma}$ marked with alphabet symbols $a_1$,$\ldots$, $a_{\sigma}$. For each  $1\le j\le \sigma$
we find the group $G_r(u_j)$ (or $G(u_j,v_j)$) that 
  contains $w_j$ and answer an LCP query $(w_j,P_j)$ on the tree induced by $G(u_j)$ (respectively $G(u_j,v_j)$). The string $P_j$ is a suffix of $P$ that satisfies $str(u,u_j)\circ P_j=a_j\circ P$. Using information in the pointer from $w$ to $u_j$, we can find $P_j$ in $O(1)$ time. 
\item[2.]
The pattern $P$ can be also located between two nodes $w'$ and $w$ of $\cT_u$ such that $str(w_r,w')$ is prefix of $P$ and $P$ is a prefix of $str(w_r,w)$. For every $j$, we follow the pointers marked with alphabet symbol $a_j$. Suppose that pointers from $w'$ and $w$ lead to locations $\ow'_j$ and $\ow_j$ respectively.
Let $w'_j$ be the lower node on the edge of $\ow'_j$ and let 
$w_j$ be the upper node on the edge of $\ow_j$. 
There are no marked nodes between $w'_j$ and $w_j$. Therefore we only need to search in the group that contains $w_j$ to complete the LCP query. 
\end{description}
%OLD TWO-CASES DESCRIPTION. DON'T NEED IT?
% We distinguish between two cases that depend on $w$ and the string $str(w_r,w)$ where $w_r$ is the root of $\cT_u$.
% \begin{description}
% \item[1.]
%   If  $str(w_r,w)$ is a prefix of $P$, we follow pointers from $w$ to alphabet nodes $w_1$, $\ldots$, $w_t$ where $1\le t\le \sigma$. For each $w_j$, $1\le j\le \sigma$, we find the group $G_r(u_j)$ (or $G(u_j,v_j)$) that 
%   contains $w_j$ and answer an LCP query $(w_j,P_j)$ on the tree induced by $G(u_j)$ (respectively $G(u_j,v_j)$). The string $P_j$ is defined so that $str(u,u_j)\circ P_j=a_j\circ P$. Using information in the pointer from $w$ to $u_j$, we can find $P_j$ in $O(1)$ time. 
% \item[2.]
% $P$ is a prefix of the  string $str(w_r,w)$. For each $j$, $1\le j\le t$, we follow pointers from $w$ to $w_j$. Let $v_j$ denote the lowest marked ancestor of $w_j$.  We answer an LCP query $(v_j, P'_j)$ on a tree induced by nodes of $G(v_j,w_j)$,
%  where $P_j$ is chosen so that $str(u,v_j)\circ P_j=a_j\circ P$. 
% \end{description}
%If $w$ is an internal node and $str(w_r,w)$ is a proper prefix of $P$, then no string $a_j\circ P$ occurs in the source text. 
The total search time is $O(\log\log n +\sigma\cdot t_{\text{small}})$ where $t_{\text{small}}$ is the time needed to answer an LCP query on a subtree of $\tau$ nodes. We  use Lemma~\ref{lemma:unrooted2}; hence $t_{\text{small}}=O(1)$. Since $\sigma=\log\log n$, a wildcard LCP query is answered in $O(\log \log n)=O(\sigma)$ time.

\section{ Wildcard Pattern Matching Queries for  $\sigma\ge \log\log n$}
\label{sec:case1}
{\bf Wildcard LCP Queries.}
We can modify the data structure of Section~\ref{sec:case0} for the case when the alphabet size $\sigma\ge \log \log n$. We divide the alphabet $\Sigma$ into groups such that every group, except the last one, contains $\log \log n$ elements. The last group contains at most $\log \log n$ elements. 
We will denote these groups $\Sigma^1$, $\ldots$, $\Sigma^g$ for $g=\lceil\,\sigma/\log \log n\,\rceil$. Instead of one wildcard tree $\cT_u$, we will store 
$g$ modified wildcard trees $\cT^1_u,\ldots,\cT^g_u$ in every node $u\in \cT$. A wildcard tree $\cT^i_u$ for a node $u$ is a compressed trie that contains all strings $s$ satisfying $a\circ s=str(u,v_l)$ for some symbol $a\in \Sigma^i$ and some marked light leaf descendant $v_l$ of $u$.  We keep the same data structure for every 
$\cT^i_u$ as in Section~\ref{sec:case0}. Thus we answer LCP queries for each group of $\log\log n$ alphabet symbols in 
$O(\log \log n)$ time.
The total time needed to answer a wildcard LCP query is $O(\lceil\,\sigma/\log\log n\,\rceil\log \log n)=O(\sigma)$.

{\bf Indexing.}
Consider a query $\tP=\phi P_1\phi P_2\ldots \phi P_d$. If $\sigma\ge \log\log n$, then our data structure for wildcard LCP queries enables us to find all occurrences of $\tP$ by answering wildcard LCP queries. We find the loci of all $a_iP_1$ for every $a_iP_1$ 
that occurs in the source text  $T$. This is achieved by answering a wildcard LCP query $(u_r,\phi P_1)$. 
For every found location $u_i^1$ we proceed as follows. If $u_i^1$ is in a middle of an edge $e$, we move one symbol down and then 
check whether the remaining symbols of an $e$ are labelled with a prefix of $P_2$. If this is the case and the remaining part of $e$ 
is labelled with $P'_2$, we answer a regular LCP query $(w_i^1,P''_2)$ such that $w_i^1$ is the node at the lower end of $e$ and $P_2=P_2'\circ P_2''$. Using the data 
structure of Bille~\etal~\cite{BilleGVV12}, an LCP query can be answered in $O(\log \log n)$ time. 
If $u_i^1$ is a node in the suffix tree, then we answer a wildcard LCP query $(u_i^1,\phi P_2)$. We continue in the same manner until
 the loci of all $xP_1\ldots xP_m$, where $x$ denotes an arbitrary symbol in $\Sigma$, are found.  A general wildcard 
pattern $\phi^{k_1}P_1\ldots\phi^{k_d}P_d$ is processed in the same 
way. 

Since the maximum number of wildcard LCP queries and standard LCP queries does not exceed $\sigma^{g}$, the total query time is $O(\sigma^g)$.
Preprocessing stage for all wildcard LCP queries takes $O(\Sigma_{i=1}^d |P_i|)$ time.  
\begin{lemma}\label{lemma:case1}
  Suppose that the alphabet size $\sigma\ge \log\log n$. Using an $O(n\log n)$-bit data structure, we can report all occurrences of a pattern $\tP=\phi^{k_1} P_1\phi^{k_2} P_2\ldots \phi^{k_d} P_d$ in $O(\sum_{i=1}^d|P_i|+\sigma^g + \occ)$ time, where $\occ$ is the number of times $\tP$ occurs in the text and 
$g=\sum_{i=1}^d k_i$.
\end{lemma}
% If we use the uncompressed suffix array and uncompressed suffix tree, the result of Lemma~\ref{lemma:case1} can be extended to the case of an arbitrarily large alphabet. The query time remains unchanged, but the space usage is 
% $O(n\log n)$ bits. 

\section{Wildcard Pattern Matching Queries for Small Alphabets}
\label{sec:case2}
In this section we consider the case when the alphabet size $\sigma< \log \log n$. We use the approach of Sections~\ref{sec:case0}
and~\ref{sec:case1}, but the notion of wildcard LCP queries is generalized.  A $t$-wildcard LCP query $(u,\tP)$ for a wildcard string $\tP=\phi^{k_1}P_1\phi^{k_2}P_2\ldots\phi^{k_d}P_d$ such that $\sum k_i=t$, finds locations of  all patterns $str(u)\circ P$, where $P=s_1s_2\ldots s_{k_1}P_1s_{k_1+1}\ldots s_{k_2}P_2\ldots s_{t-1}s_tP_d$ and $s_i$, $1\le i\le t$, are  arbitrary alphabet symbols, in the suffix tree. A $1$-wildcard LCP query, used in the previous sections, takes $O(\log \log n)$ time and can replace up to $\sigma$ standard wildcard queries. Hence, when the alphabet size $\sigma$ is small, we cannot achieve noteworthy speed-up in this way. A $t$-wildcard LCP query can replace up to $\sigma^t$ regular LCP queries and lead to more significant  speed-up even when $\sigma$ is very small. \no{The parameter $t$ of the $t$-wildcard subtree will be selected in such a way that a $t$-wildcard LCP query replaces $O(\sigma^t/\log\log n)$ regular LCP queries. the 
number of answered wildcard LCP queries (we will show below that it equals $2^t$)
does not exceed $O(\sigma^t/\log \log n)$. }
We will use iterated  wildcard subtrees in order to support $s$-wildcard LCP queries efficiently.  Our construction consists of two parts.  We mark selected nodes in the suffix tree $\cT$ and  divide it into subtrees $\cT_i$ of size $O(\tau_1)$; we keep a data structure that supports $t_1$-wildcard LCP queries on the subtree $\cT^m$ induced by marked nodes of $\cT$. We also mark selected nodes, further called secondary marked nodes, in each subtree $\cT_i$ and divide $\cT_i$ into $\cT_{i,j}$ of size $O(\tau_2)$. Let $\cT^m_i$ be the subtree induced by secondary 
marked node of $\cT_i$; we keep a data structure that answers standard wildcard LCP queries on $\cT_i^m$.  
 Details of our data structure and parameter values can be found below. 

{\bf Trees $\cT_i$ and $\cT^m$.}
Let $t_1=\log_{\sigma/2}\log\log n$ and $\tau_1=\sigma^{t_1}\log^{t_1+1}n$. We use the same scheme as in Section~\ref{sec:lessspace} to mark every $\tau_1$-th leaf and selected internal nodes, so that the suffix tree $\cT$ is divided into subtrees $\cT_i$ of size $O(\tau_1)$ and the number of marked nodes is $O(n/\tau_1)$. Trees $\cT_i$ correspond to groups $G_j(u)$ and $G(u,v)$ defined 
in section~\ref{sec:lessspace}.

Let $\cT^m$ be the tree induced by marked nodes. We iteratively augment $\cT^m$ with wildcard subtrees. For any marked internal 
node $u$, the (level-$1$) wildcard subtree $\cT_u$ is a compressed trie containing all strings $s$, such that $a\circ s=str(u,v_l)$ for some marked  light leaf descendant  $v_l$ of $u$. 
We also keep a level-$(i+1)$ wildcard subtree $\cT_{w}$ for every node $w$ in a level-$i$ wildcard subtree $\cT_u$. $\cT_w$ contains all strings $s$ such that $a\circ s=str(u,v_l)$  for some alphabet symbol $a$ and a light leaf descendants $v_l$ of $w$. We construct level-$i$ wildcard subtrees for $1\le i \le t_1$.  The parameter $t_1$ is chosen in such way that $\sigma^{t_1}=2^{t_1}\log\log n$ and $t=\log_{\sigma}\log \log n$. 
Every node in all level-$i$ wildcard trees has pointers to 
the corresponding locations in the alphabet tree $\cT$. 
Each pointer also contains information about the subtree $\cT_i$

The total number of nodes and pointers in wildcard subtrees is 
$(n/\tau_1)\sigma^{t_1}\log^{t_1}n$.  Level-$t$ wildcard subtrees can be used  to answer unrooted $t$-wildcard LCP queries on $\cT_m$ in $O(2^t\log \log n)$ time; our method is quite similar to the procedure for answering wildcard queries in~\cite{ColeGL04}.  Consider a query $(\ou,\tP)$, where $\ou$ is a location in the alphabet tree or in some $i$-wilcard subtree. We distinguish between the following four cases.
(i) If $\ou$ is on a tree edge and the next symbol is a wildcard,
 we simply move down by one symbol along that edge. 
(ii) Suppose that $\ou$ is on a tree edge $e$ and the next symbols are a string $P_n$ of alphabet symbols. Let $l$ denote the string label of the part of $e$ below $\ou$, $l=str(\ou,u')$ where $u'$ is the lower node on $e$.  
We compute $o=LCP(P_n,l)$. and move down by $\min(|l|,o)$ symbols along $e$. 
(iii) If $\ou$ is a node and the next unprocessed symbol in $\tP$
is a wildcard, our procedure branches and visits two locations:
we move down by one symbol along the edge to the heavy child of $\ou$ and visit the root of the wildcard tree $\cT_{\ou}$ (if $\ou$
is on a level-$i$ wildcard tree, we visit the root of the $(i+1)$-subtree $\cT_{\ou}$). 
(iv) If $\ou$ is a node and the next symbols are a string $P_n$ of alphabet symbols, we answer a standard LCP query $(\ou,P_n)$. 
The procedure is finished when we cannot move down from any location that is currently visited. 
The number of branching points is $2^t$ and we answer $2^t$ 
standard LCP queries. 
We need $O(\sigma^t)$ time to return from locations in wildcard trees to the corresponding locations in the alphabet tree. 
Thus the total time is $O(2^t\log \log n+\sigma^t)=O(\sigma^t)$.
When the search in $\cT^m$ is completed we can continue searching in subtrees $\cT_j$. %\longver{Data structures for subtrees $\cT_i$ are described in  
%Section~\ref{sec:subtrees}, where we show that an unrooted LCP query on $\cT_i$ can be answered in $O((\log^{(3)}n)^{1/2})$ time.} 
\shortver{Data structures for subtrees $\cT_i$ are described in~\cite{arxivver}, where we show that an unrooted LCP query on $\cT_i$ can be answered in $O((\log^{(3)}n)^{1/2})$ time.}

\longver{
{\bf Data Structures for Subtrees $\cT_i$}
Let $\cT_i$ be a subtree of the alphabet tree $\cT$. We set  
$\tau_2= \log^2 n$. Again, we mark $O(n/\tau_2)$ nodes in $\cT_i$, so that $\cT_i$ is divided into $O(n/\tau_2)$ subtrees $\cT_{i,j}$. 
Marked nodes in $\cT_i$ will be called secondary marked nodes. 
Let  $\cT^m_i$ denote the subtree of $\cT_i$ induced by secondary marked nodes. We keep a data structure that answers standard  LCP queries on $\cT^m_i$. 
This data structure is the same as the data structure for $\cT^m$. But standard LCP queries on $\cT^m_i$ and its  wildcard trees can be answered in $\mu(n)=O(\sqrt{\log \tau_1})=O(\sqrt{\log \log \log n})$ time\footnote{In fact, a slightly better time $O(\sqrt{\log^{(3)}n/\log^{(4)}n})$ can be achieved. We use this slightly worse time to simplify the final Theorem.}; see Lemma~\ref{lemma:unrootedlarge} in Section~\ref{sec:hpath}. 
Finally, we store a data structure of Lemma~\ref{lemma:unrooted3} for each subtree $\cT_{i,j}$.  
Since we also keep a suffix array with $\tsa=O(1)$, we can answer LCP queries on $\cT_{i,j}$ in $O(1)$ time. 
We can use the combination of $\cT_i^m$ and subtrees $\cT_{i,j}$ 
to answer LCP queries on $\cT_i$ in $O((\log^{(3)}n)^{1/2})$ time. 
}

{\bf Wildcard String Matching.}
It follows from the above description that we can answer $t_1$-wildcard LCP queries in $O(\sigma^{t_1}\sqrt{\log^{(3)}n})$ time.  
Consider now an arbitrary pattern $\tP=\phi^{k_1}P_1\phi^{k_2}P_2\ldots \phi^{k_d}P_d$. We divide it into chunks $\tP[1]$, $\tP[2]$, $\ldots$, 
$\tP[r]$, such that each chunk $\tP[i]$, $i\ge 2$, contains exactly $t_1$ wildcard symbols. The chunk $P[1]$ contains  $v\le t_1$ wildcard symbols. 
 
We start at the root and find locations of all $\tP[1]=\phi^{k_1}P_1\ldots \phi^{k_f}P_f\phi^r$ where $r\le k_{f+1}$. If $\sum _{i=1}^f |P_f| > (\log \log n)\cdot \sigma^{t}$, we
answer at most $\sigma^t$ standard LCP queries in $O(\sigma^t\log\log n)=O(\sum_{i=1}^f|P_i|)$ time. 
If $\sum_{i=1}^f |P_i| \le (\log \log n)\cdot \sigma^t$, then the total length of $\tP[1]$ is at most $\ell=(\log \log n)\cdot \sigma^t+ t$. Since $\sigma<\log \log n$, there are 
$O((\log \log n)^{\ell})$ different patterns and each of this patterns fits into one machine word. Hence, all string patterns $P_s$ that match $\tP[1]$ can be generated in $O(\sigma^v)$ time. We keep a look-up table with locations of all strings $P$, such that $|P|\le \ell$ in $\cT$. Using this table we find locations of all $P_s$ that match $\tP[1]$ and occur in the source text. 
For every such location $\ou$, we answer queries $(\ou_1,\tP[2])$, $(\ou_2,\tP[3])$, $\ldots$, where $\ou_1=\ou$ and $\ou_i$ for $i>1$ is an answer to some query $(\ou_{i-1},\tP[i])$. It is easy to show that the total query time is $O(\sum_{i=1}^d|P_i| +\sigma^g\sqrt{\log^{(3)}n}+\occ)$.
\begin{theorem}
\label{theor:fast}
  If the alphabet size $\sigma=O(1)$ and $\sigma>2$, then there exists an $O(n\log^{\eps}n)$-bit data structure that  reports all $\occ$ occurrences of a wildcard pattern  $\phi ^{k_1}P_1\phi ^{k_2}P_2\ldots \phi ^{k_d}P_d$ in $O(\sum_{i=1}^d|P_i| +\sigma^g\sqrt{\log^{(3)}n}+\occ)$ time.
\end{theorem}
We remark that the same query time as in Theorem~\ref{theor:fast}
can be also achieved for a non-constant $\sigma$; the space usage would grow to $O(n\log n)$ bits, however. To obtain 
this result, we would need to use  standard (uncompressed) suffix tree and suffix array for the source data.

\subparagraph*{Acknowledgement}
The second author wishes to thank Gonzalo Navarro for pointing him to~\cite{Rao02}.

\bibliographystyle{abbrv} %{plain}
\bibliography{wildcards}

\longver{
\newpage
\appendix 
\section{Auxiliary Data Structures for Unrooted LCP Queries}
\label{sec:hpath}
\subparagraph*{A Compact Data Structure for Heavy-Path Decomposition.}
Let $\cT$ denote a subtree of the suffix tree  induced by $f=O(\log^3n)$ 
consecutive suffixes. 
%We can find for any node $u\in \cT_0$, the 
%leaf $v_h$ such that $u$ and $v_h$ are on the same path $h_j$ 
%in the heavy path decomposition of $\cT_0$. We need only $O(f)$ additional bits %to find $v_h$ for any $u\in\cT_0$ in $O(1)$ time. 

We mark every $\tau'$-th leaf of $\cT$ for a parameter $\tau'=\log \log n$. Then we mark internal nodes and all nodes of $\cT$ are divided into groups in the same way as in Section~\ref{sec:lessspace}.  For every group we store 
its topology in $O(\lg \lg n)$ bits. Hence, we can read the data about a group into one machine word. Using a look-up table of size $o(n)$, we can find 
the heavy path of any node $v$ such that $v$ is not marked and the leaf $v_h$ on that path.   For every marked node $u_m$ we explicitly store the index of the leaf $v_h$ that is on the same heavy path as $u_m$.  There are $O(f/\log\log n)$ marked nodes and each node in $\cT$ can be specified with $O(\log\log n)$ bits.
Thus we need $O(f)$ bits for all marked nodes. Hence, we can determine the heavy path of any node $u\in \cT$ in $O(1)$ time using $O(f)$ additional bits. 
We recall that a data structure $H_j$ uses $O(\log\log n)$ bits per node.
\subparagraph*{Proof of Lemma~\ref{lemma:unrooted4}.} 
\begin{proof}
We slightly modify the data structures $D(u)$ stored in the nodes of $\cT$. If $S(u)$ contains at most $(\log \log n)^2$ elements, 
then $R(u)$ is discarded. We can simply find any suffix of $S(u)$ and compare it to $P_j$ in $O(\tsa)$ time per suffix. 
Using binary search, we can find the predecessor of $P_j$ in $S(u)$ in $O(\tsa\cdot\log\log \log n)$ time.  
If $|S(u)|>(\log\log n)^2$, we select every $(\log \log n)^2$-th element of $S(u)$ and keep them in a set $S'(u)$. We maintain 
$D(u)$ on the ranks of elements in $S'(u)$. To find a predecessor of $P_j$ in $S(u)$ we first find its predecessor in $S'(u)$ 
using $D(u)$. When its predecessor in $S'(u)$ is known, we can search among $(\log \log n)^2$ consecutive suffixes as described 
above. 

We also use the same technique to reduce the space usage of data structures $H_j$. Recall that $H_j$ finds for any $d_q$ the lowest node $u_q$ on the heavy path $h_j$, such that the depth of $u$ does not exceed $d_q$.
We select every $(\log\log n)$-th node on $h_j$ and store the depths of selected 
nodes in the data structure $H_j$ implemented using Lemma~\ref{lemma:small}.   
Instead of $H_j$, we keep a data structure $H'_j$ that 
contains the string depths of every $\log\log n$-th node on a heavy path $h_j$. All $H_j$ need $O((f/\log\log n)\log \log n)=O(f)$ bits. 
To find the lowest node of depth at most $d_q$ on a path $h_j$, we find the predecessor $d_e$ of $d_q$ in $H_j$. Let $u_1$  be the node of depth $d_e$ 
on $h_j$ and let $u_2$ be the next node whose depth is stored in $H_j$.
Nodes $u_1$ and $u_2$ can be found in $O(\tsa)$ time using $H_j$. 
The node $u_q$ is between $u_1$ and $u_2$ and can be found in $O(\tsa\cdot \\log^{(3)}n)$ time by binary search.   
The total time to answer an unrooted LCP query is dominated by searching for predecessor in $S(u)$ and $H(u)$.
\end{proof}

\subparagraph*{LCP Queries on Large Sets}
The approach of  section~\ref{sec:smallset} can be also used to obtain a data structure that answers queries on an arbitrarily large set of suffixes in $O(\log \log n)$ time.  Let $\cT_1$ denote the subtree of the suffix tree $\cT$ induced by suffixes from a set 
$S$. Unrooted LCP queries $(u,P)$ for $u\in\cT_1$ can be answered in $O(\min(\log \log n, \sqrt{\log f/\log \log n}))$ time for $f=|S|$. 
\begin{lemma}
\label{lemma:unrootedlarge}
Let $S$ be a set of $f$ suffixes of a text $T$. There exists an $O(|S|\log^2 n)$-bits data structure that answers unrooted LCP 
queries on a subtree induced by $S$ in time  $O(\min(\log \log n, \sqrt{\log f/\log \log n}))$. 
\end{lemma}
\begin{proof}
 We consider the heavy path decomposition of $\cT_1$ and keep data structures $H_j$ and $D(u)$ defined in the proof of Lemma~\ref{lemma:unrooted2}. Since $S$ can be large, we implement $H_j$ and $D(v)$ as van Emde Boas data structures~\cite{veb} or using the result from~\cite{BeameF02} so that predecessor queries are answered 
in $O(\min(\log \log n, \sqrt{\log f/\log \log n}))$ time. The total number of elements in 
all $H_j$ and all $D(v)$ is $O(n)$ and $O(n\log n)$ respectively. Since each $H_j$ and $D(v)$ uses linear space, the total space usage is $O(n\log n)$ words of $\log n$ bits.  
\end{proof}

% \section{Data Structures for Subtrees $\cT_i$}
% \label{sec:subtrees}
% Let $\cT_i$ be a subtree of the alphabet tree $\cT$. We set  
% $\tau_2= \log^2 n$. Again, we mark $O(n/\tau_2)$ nodes in $\cT_i$, so that $\cT_i$ is divided into $O(n/\tau_2)$ subtrees $\cT_{i,j}$. 
% Marked nodes in $\cT_i$ will be called secondary marked nodes. 
% Let  $\cT^m_i$ denote the subtree of $\cT_i$ induced by secondary marked nodes. We keep a data structure that answers standard  LCP queries on $\cT^m_i$. 
% This data structure is the same as the data structure for $\cT^m$. But standard LCP queries on $\cT^m_i$ and its  wildcard trees can be answered in $\mu(n)=O(\sqrt{\log \tau_1})=O(\sqrt{\log \log \log n})$ time\footnote{In fact, a slightly better time $O(\sqrt{\log^{(3)}n/\log^{(4)}n})$ can be achieved. We use this slightly worse time to simplify the final Theorem.}; see Lemma~\ref{lemma:unrootedlarge} in Section~\ref{sec:hpath}. 
% Finally, we store a data structure of Lemma~\ref{lemma:unrooted3} for each subtree $\cT_{i,j}$.  
% Since we also keep a suffix array with $\tsa=O(1)$, we can answer LCP queries on $\cT_{i,j}$ in $O(1)$ time. 
% We can use the combination of $\cT_i^m$ and subtrees $\cT_{i,j}$ 
% to answer LCP queries on $\cT_i$ in $O((\log^{(3)}n)^{1/2})$ time. 
}

\end{document}